\documentclass[12pt, letterpaper]{article}

\usepackage{arxiv}

\usepackage[utf8]{inputenc}
\usepackage[T1]{fontenc}   % use 8-bit T1 fonts
\usepackage{amsfonts}       % blackboard math symbols

\usepackage[small]{caption}
\usepackage{graphicx}
\usepackage{subfigure}
\usepackage{amsmath}
\usepackage{amsthm}
\usepackage{amssymb}
\usepackage{booktabs}
\usepackage[algo2e,ruled,vlined]{algorithm2e}%for algorithm
\usepackage{algorithm}
\usepackage{algorithmic}
\usepackage{color}

\usepackage[square,numbers]{natbib}
\usepackage{authblk}

\usepackage[pdftex,
pdfauthor={Mengjing Chen; Yang Liu; Weiran Shen; Yiheng Shen; Pingzhong Tang; Qiang Yang},
pdftitle={Mechanism Design for Multi-Party Machine Learning},hidelinks]{hyperref}

\title{Mechanism Design for Multi-Party Machine Learning 
}
\author[1]{Mengjing Chen\thanks{ccchmj@qq.com}\enspace}
\author[2]{Yang Liu\thanks{yangliu@webank.com}\enspace}
\author[3]{Weiran Shen\thanks{emersonswr@gmail.com}\enspace}
\author[1]{Yiheng Shen\thanks{yiheng\_shen@outlook.com}\enspace}
\author[1]{Pingzhong Tang\thanks{kenshinping@gmail.com}\enspace}
\author[4]{Qiang Yang\thanks{qyang@cse.ust.hk}\enspace}
\affil[1]{Tsinghua University}
\affil[2]{WeBank}
\affil[3]{Carnegie Mellon University}
\affil[4]{Hong Kong University of Science and Technology}

\newtheorem{example}{Example}
\newtheorem{theorem}{Theorem}
\newtheorem{definition}{Definition}
\newtheorem{observation}{Observation}
\newtheorem{corollary}{Corollary}

\newtheorem{assumption}{Assumption}
\DeclareMathOperator*{\argmax}{arg\,max}
\DeclareMathOperator*{\argmin}{arg\,min}
\newcommand{\wel}{\textsc{Wel}}
\newcommand{\rev}{\textsc{Rev}}

\begin{document}
 \maketitle
\begin{abstract}
 In a multi-party machine learning system, different parties cooperate on optimizing towards better models by sharing data in a privacy-preserving way. A major challenge in learning is the incentive issue. For example, if there is competition among the parties, one may strategically hide his data to prevent other parties from getting better models.
	
%In this paper, we study the problem through the lens of mechanism design. Compared with the standard mechanism design setting, our setting has several fundamental differences. First, each agent's valuation has externalities that depend on others' true types. We call this setting \emph{mechanism design with type-imposed externalities}. Second, each agent can only misreport a lower type, but not the other way round. We show that some results (e.g., the truthfulness of the VCG mechanism) in the standard mechanism design setting fail to hold. 
In this paper, we study the problem through the lens of mechanism design and incorporate the features of multi-party learning in our setting. First, each agent's valuation has externalities that depend on others' types and actions.  Second, each agent can only misreport a type lower than his true type, but not the other way round. We call this setting \emph{interdependent value with type-dependent action spaces}. We provide the optimal truthful mechanism in the quasi-monotone utility setting. We also provide necessary and sufficient conditions for truthful mechanisms in the most general case. Finally, we show the existence of such mechanisms is highly affected by the market growth rate and provide empirical analysis.
%Third, the learned model may depend on the agents' reports, while in the standard auction setting, the item being sold does not depend on the bidders' bids. 
%In this paper, we study  mechanism with limited action space and type-imposed externality. Agents cannot overreport and the valuation of each agent depends both on the other agents' true types and allocation outcome. Such differences with the standard mechanism design prevents some results (e.g., the VCG mechanism) applying directly to the problem.

\end{abstract}

\section{Introduction}\label{sec:inro}
%Due to limited computational resources and data size, a single data owner may not be able to train a model with very high quality. 
In multi-party machine learning, a group of parties cooperates on optimizing towards better models. This concept has attracted much attention recently \cite{hu2019fdml,shokri2015privacy,smith2017federated}. The advantage of this approach is that, it can make use of the distributed datasets and computational power to learn a powerful model that anyone in the group cannot achieve alone.

To make multi-party machine learning practical, a large body of works focus on preserving data privacy in the learning process \cite{abadi2016deep,yonetani2017privacy,shokri2015privacy}. However, the incentive issues in the multi-party learning have largely been ignored in most previous studies, which results in a significant reduction in the effectiveness when putting their techniques into practice. Previous works usually let all the parties share the same global model with the best quality regardless of their contributions. This allocation works well when there are no conflicts of interest among the parties. For example, an app developer wants to use the users' usage data to improve the user experience. All users are happy to contribute data since they can all benefit from such improvements \cite{mcmahan2017communication}. 

When the parties are competing with one another, they may be unwilling to participate in the learning process since their competitors can also benefit from their contributions. Consider the case where companies from the same industry are trying to adopt federated learning to level up the industry's service qualities. Improving other companies' services can possibly harm their own market share, especially when there are several monopolists that own most of the data.

Such a cooperative and competitive relation poses an interesting challenge that prevents the multi-party learning approach from being applied to a wider range of environments. In this paper, we view this problem from the multi-agent system perspective, and address the incentive issues mentioned above with the mechanism design theory.
%To model such cooperative and competitive relation among agents, we assume the valuation of each agent is type-dependent, meaning that agents care about both the allocation outcome and the other's true types.

%\subsection{Mechanism Design with Type-Imposed Externality}
Our setting is a variant of the so-called interdependent value setting \cite{Milgrom1982Competitive}. A key difference between our setting and the standard interdependent value setting is that each agent cannot ``make up'' a dataset that is of higher quality than his actual one. Thus the reported type of an agent is capped by his true type. We call our setting \emph{interdependent value with type-dependent action spaces}. The setting that agents can never over-report is common in practice. One straightforward example is that the sports competitions where athletes can show lower performance than their actual abilities but not over-perform. The restriction on the action space poses more constraints on agents' behaviors, and allows more flexibility in the design space.

%Consider an example with two companies, where company A owns much more data than company B. Suppose A only contributes a small amount of data, but instead of using the model learned via the federated learning framework, A uses the model trained based on all his data. In this case, the actual amount of data that A has may affect B's market share, due to the competition between them. 

%Compared to the standard mechanism design setting, the restriction on the action space poses more constraints on the agent behaviors, and allows more flexibility in the design space. Therefore, some results in the standard mechanism design setting may no longer hold. Another key difference is that the valuation of each agent is a function of the allocation outcome and the actual types of all agents, while in the standard mechanism design setting, the value function only depends on what he gets and his own type. Recall in the competing companies example, the market share of a company depends on the models that other companies use to serve their customers and the models they use may not be the same even if they get a global model from the learning platform, as they may have incentives to hide some data and improve the global model afterward with them.

We first formulate the problem mathematically, and then apply techniques from the mechanism design theory to analyze it. Our model is more general than the standard mechanism design framework, and is also able to describe other similar problems involving both cooperation and competition. 

We make the following contributions in this paper:
\begin{itemize}
	\item We model and formulate the mechanism design problem in multi-party machine learning, and identify the differences between our setting and the other mechanism design settings.
	%\item We propose a implementable learning protocol for the multi-party learning to accompany our mechanism.
	\item For the quasi-monotone externalities setting, we provide the revenue-optimal and truthful mechanism. For the general valuation functions, we provide both the necessary and the sufficient conditions for all truthful and individually rational mechanisms.
	\item We analyze the influence of the market size on mechanisms. When the market grows slowly, there may not exist a mechanism that achieves all the desirable properties we focus on.	
\end{itemize}

\subsection{Related Works}
%{\color{red} add a paragraph on classic papers and recent advances in auction design}
A large body of literature studies mechanisms with interdependent values \cite{Milgrom1982Competitive}, where agents’ valuations depend on the types of all agents and the intrinsic qualities of the allocated object. \cite{RoughgardenT16} extend Myerson’s auction to specific interdependent value settings and characterize truthful and rational mechanisms. They consider a bayesian setting while we do not know any prior information. \cite{ChawlaFK14} propose a variant of the VCG auction with reserve prices that can achieve high revenues. They consider value functions that are single-crossing and concave while we consider environments with more general value functions. \cite{Claudio2004Efficiency}  give a two-stage Groves mechanism that guarantees truthfulness and efficiency. However, they require agents to report their types and valuations before the final monetary transfer are made while in our model, agents can only report their types.

%Our paper is part of the body of literature on mechanisms with interdependent values \cite{Milgrom1982Competitive}, where agents’ valuations depend on the types of all agents and the intrinsic qualities of the allocated object. \cite{RoughgardenT16} extend Myerson’s auction to specific interdependent value settings and characterize ex post IC and IR mechanisms. \cite{ChawlaFK14} propose a variant of VCG auction with reserve prices which has a constant approximation to the maximal revenue.  \cite{Claudio2004Efficiency}  give a  two-stage Groves mechanism that guarantees IC and efficiency. However, they require agents report their types and valuations before the final monetary transfer are made while in our model, agents only need to report their types. Despite of large of research on mechanisms with interdependent values, there are three differences between our setting and the standard interdependent value settings that prevent the mechanisms proposed above from applying to the multi-parity learning model. First, the set of available outcomes is correlated with agents’ types in our setting while the outcome set is fixed. Second, even if an agent does not participate in the mechanism, he can never quit the market. The agent would still compete with each other in the market and influence the other agents’ valuations. However, to the best of our knowledge, no previous work has considered this situation.  At last, agents in the multi-parity learning can use their private models to compete instead of the ones that the system allocates to them. 

In our setting, agents have restricted action spaces, i.e., they can never report types exceeding their actual types. There is a series of works that focus on mechanism design with a restricted action space \cite{blumrosen2006implementation,blumrosen2002auctions,ausubel2004efficient}. The discrete-bid ascending auctions  \cite{david2007optimal,chwe1989discrete,ausubel2004efficient} specify that all bidders' action spaces are the same bid level set. Several works restrict the number of actions, such as bounded communications \cite{blumrosen2002auctions}. Previous works focus on mechanisms with independent values and discrete restricted action spaces, while we study the interdependent values and continuous restricted action spaces setting.

%The allocation of each agent has external effects on the profits of other agents in our multi-party learning setting. 
%The external effect on agents' valuations aims to model competitions and cooperations among agents. A vast literature have studied mechanisms with externalities \cite{jehiel1996not,deng2011money,jehiel1999multidimensional}, but none of them consider the external effects that are related to other agents' actual types.
% For example, \citeauthor{jehiel1996not}\shortcite{jehiel1996not} propose an optimal auction for selling one item when potential buyers have negative externalities in the full information setting. \citeauthor{deng2011money}\shortcite{jehiel1996not} follow their work and exploit auctions with negative externalities for multiple items. 
%Our setting is different from theirs, as external effects are related to other agents' actual types.
%They both focus on negative externailities while the externality effect can be possitive in our setting.%The externalities can also describe the values of goods or services after consumption.  
%digital goods

The learned model can be copied and distributed to as many agents as possible, so the supply is unlimited. A line of literature focuses on selling items in unlimited supply such as digital goods \cite{goldberg2001competitive,goldberg2003envy,fang2016digital}. However, the seller sells the same item to buyers while in our setting we can allocate models with different qualities to different agents.

%Although we design allocations and payments for agents as auctions, there are two major differences between our mechanism and auction models \cite{krishna2009auction}. The auction model assumes that agents can bid any types while in our model agents are not allowed to ``overbid''. The second difference is the utility function. We adopt a more general form of utility function, where each agent's valuation is about the outcome of allocations and actual types of all agents and while in auctions agents only care about their own allocations.

\cite{redko2019fair} study the optimal strategies of agents for collaborative machine learning problems. Both their work and ours capture the cooperation and competition among the agents, but they only consider the case where agents reveal their total datasets to participate while agents can choose to contribute only a fraction in our setting. \cite{kang2019incentive} study the incentive design problem for federated learning, but all their results are about a non-competitive environment, which may not hold in real-world applications. 

Another closely related topic is the incentives in machine learning problems, such as strategyproof classification \cite{meir2012algorithms,meir2009strategyproof}, incentive compatible linear regression \cite{cummings2015truthful,dekel2008incentive,perote2004strategy}. The main objective of these mechanisms is to get an unbiased learning model with strategic agents and all agents receive the results from the same model. We also consider the collaborative learning setting, but different models are distributed to different agents.

%There is also a line of works on distributed algorithmic mechanism design \cite{nisan2008agt,feigenbaum2004distributed,nisan2001algorithmic}. They focus on distributed mechanisms under the condition that the untrusted center cannot be used.  We also analyze mechanisms that are distributed,  but the center is trusted in our setting.

\section{Preliminaries}\label{sec:pre}
In this section, we introduce the general concepts of mechanism design and formulate the multi-party machine learning as a mechanism design problem. A multi-party learning consists of a central platform and several parties (called agents hereafter). The agents serve their customers with their models trained using their private data. Each agent can choose whether to enter the platform. If an agent does not participate, then he trains his model with only his own data. The platform requires all the participating agents to contribute their data in a privacy-preserving way and trains a model for each participant using a (weighted) combination of all the contributions. Then the platform returns the trained models to the agents. 

%We assume that all agents, no matter whether they enter the platform or not, use the same model structure. 
We assume that all agents use the same model structure. 
Therefore, each participating agent may be able to train a better model by making use of his private data and the model allocated to him. One important problem in this process is the incentive issue. For example, if the participants have conflicts of interest, then they may only want to make use of others' contributions but are not willing to contribute with all their own data. To align their incentives, we allow the platform to charge the participants according to some predefined rules. 
%In this paper, we analyze this problem from the angle of mechanism design. 

Our goal is to design allocation and payment rules that encourage all agents to join the multi-party learning as well as to contribute all their data.

\subsection{Valid Data Size (Type)}
Suppose there are $n$ agents, denoted by $N=(1,2,\dots,n)$, and each of them has a private dataset $D_i$ where $D_i\cap D_j=\emptyset, \forall i\ne j$. For ease of presentation, we assume that a model is fully characterized by its quality $Q$ (e.g., the prediction accuracy), and the quality only depends on the data used to train it. We have the following observation:

\begin{observation}
	If the agents could fake a dataset with a higher quality, any truthful mechanism would make agents gain equal final utility.
\end{observation}
Suppose that two agents have different true datasets $D_1$ and $D_2$. We assume that all other agents truthfully report datasets $D_{-i}$. If truthfully reporting dataset $D_1$ and $D_2$ finally leads to different utility, w.l.o.g, we let $u(D_1,D_{-i})<u(D_2,D_{-i})$, then if an agent has true dataset $D_1$, he would report $D_2$ and use the dataset allocated by the platform in the market. All his behavior is the same as that of an agent with real dataset $D_2$. Thus if a mechanism is truthful, any reported dataset would lead to the same final utility and thus it is pointless to discuss the problem. Hence, we make the assumption that the mechanism is able to identify the quality of any dataset. All agents can only report a dataset with a lower quality.

For simplicity, we measure the contribution of a dataset to a trained model by its \emph{valid data size}. Thus we have the following assumption:
\begin{assumption}\label{asm:type}
	The model quality $Q$ is bounded and monotone increasing with respect to the valid data size $s \ge 0$ of the training data:
%	\begin{gather*}
%	\textstyle
%		\text{(1)}\  Q(0)=0\  \text{and}\  Q(s)\le 1, \ \forall s ; \ \text{(2)}\  Q(s')> Q(s),\  \forall s'> s.
%	\end{gather*}	
%	(1) $Q(0)=0$ and $Q(s)\le 1$, $\forall s $; (2) $Q(s')> Q(s)$, $\forall s'> s $.
	\begin{enumerate}
		\item $Q(0)=0$ and $Q(s)\le 1$, $\forall s $;
	    \item $Q(s')> Q(s)$, $\forall s'> s $.
	\end{enumerate}
\end{assumption}
The valid data size of every contributor's data is validated by the platform in a secure protocol (which we propose in the full version ). Let $t_i \in \mathbb{R}_+$ be the valid data size of agent $i$ 's private dataset $D_i$. We call $t_i$ the agent's \emph{type}. The agent can only falsify his type by using a dataset of lower quality (for example, using a subset of $D_i$, or adding fake data), which decreases the contribution to the trained model as well as the size of valid data. As a result, the agent with type $t_i$ cannot contribute to the platform with a dataset higher than his type:

\begin{assumption}\label{asm:action_space}
	Each agent $i$ can only report a type lower than his true type $t_i$, i.e., the action space of agent $i$ is $[0, t_i]$.
\end{assumption}

\subsection{Mechanism}
Let $t=(t_1,t_2,\dots,$$t_n)$ and $t_{-i}=(t_1,\dots,t_{i-1},t_{i+1},\dots,t_n)$ be the type profile of all agents and all agents without $i$, respectively. Given the reported types of agents, a mechanism specifies a numerical allocation and payment for each agent, where the allocation is a model in the multi-party learning. Formally, we have: 

\begin{definition}[Mechanism]%aamas template balance problem
	A mechanism $\mathcal{M}=(x,p)$ is a tuple, where
		\item $x=(x_1,x_2,\cdots,x_n)$, where $x_i$: $\mathbb{R}_+^n$$\mapsto \mathbb{R}$ takes the agents' reported types as input and outputs the model quality for agent $i$;
		\item $p=(p_1,p_2,\cdots,p_n)$, where $p_i$: $\mathbb{R}_+^n$$\mapsto \mathbb{R}$ takes the agents' reported types as input and specifies how much agent $i$ should pay.
%	\begin{itemize}
%		\item $x=(x_1,x_2,\cdots,x_n)$, where $x_i$: $\mathbb{R}_+^n\mapsto \mathbb{R}$ is the allocation function for agent $i$, which takes the agents' reported types as input and decides the model quality for agent $i$ as output;
%		\item $p=(p_1,p_2,\cdots,p_n)$, where $p_i$: $\mathbb{R}_+^n\mapsto \mathbb{R}$ is the payment function for agent $i$, which takes the agents' reported types as input and specifies how much agent $i$ should pay to the mechanism.
%	\end{itemize}
\end{definition}
% \begin{definition}
% A mechanism $\mathcal{M}=(x,p)$ is a tuple, where
% \end{definition}
In a competitive environment, a strategic agent may hide some of data and does not use the model he receives from the platform. Thus the final model quality depends on both the allocation and his actual type. We use valuation function $v_i(x(t'),t)$ to measure the profit of agent $i$.
\begin{definition}[Valuation]%aamas template balance problem
We consider valuation functions $v_i(x(t'),t)$ that depend not only on the allocation outcome $x(t')$ where $t'$ is the reported type profile, but also on the actual type profile $t$.
\end{definition}
%Note that type-dependent valuation is quite different from the valuation defined in the standard mechanism design setting which is only related with his own allocation and true type. In the rest of our paper, the type-dependent valuation is called valuation for short. 
We assume the model agent $i$ uses to serve customers is:
\begin{gather*}
\textstyle
    q_i=\max\{x_i(t') , Q(t_i)\},
\end{gather*}
%$$q_i=\max\{x_i(t') , Q(t_i)\},$$
%\com{that is to either use the model trained by his own data, or use the model allocated by the mechanism}
where $Q(t_i)$ is the model trained with his own data. The valuation of agent $i$ depends on the final model qualities of all agents due to their competition. Hence $v_i$ can also be expressed as $v_i(q_1,\dots,q_n)$.

%With the secure multi-party computation, all the agents could not use the models allocated to them and their hidden data to train a better model. If an agent hides some data and report untruthfully, he could only select one from the two models (the model allocated by the platform and the one trained by his all his dataset). We consider $v_i$ depends on the model qualties that agents.Hence we consider $v_i$ can be expressed as $v_i(\max\{x(t_i)\})$

We make the following assumption on agent $i$'s valuation:
\begin{assumption}\label{asm:value_monotone}
 Agent $i$'s valuation is monotone increasing with respect to true type $t_i$ when the outcome $x$ is fixed.
\begin{gather*}
\textstyle
v_i(x,t_i,t_{-i})\ge v_i(x,\hat{t}_i,t_{-i}), \forall x, \forall t_i\ge \hat{t}_i, \forall t_{-i}, \forall i.
\end{gather*}
\end{assumption}
This is because possessing more valid data will not lower one's valuation. Otherwise, an agent is always able to discard part of his dataset to make his true type $t'_i$. Suppose that each agent $i$'s utility $u_i(t,t')$ has the form: 
\begin{gather*}
u_i(t,t')=v_i(x(t'),t)-p_i(t'),
\end{gather*}
where $t$ and $t'$ are true types and reported types of all agents respectively.
As mentioned above, an agent may lie about his type in order to benefit from the mechanism. The mechanism should incentivize truthful reports to keep agents from lying.

\begin{definition}[Incentive Compatibility (IC)]
	A mechanism is said to be incentive compatible, or truthful, if reporting truthfully is always the best response for each agent when the other agents report truthfully:
	\begin{gather*}
	\textstyle
     u_i(x(t_i,t_{-i}),t)\ge u_i(x_i(t'_i,t_{-i}),t), \forall t_i \ge t_i',\forall t_{-i}, \forall i.
	\end{gather*}
	
\end{definition}

% \begin{definition}[Dominant-strategy incentive compatibility (DSIC)]
% 	A mechanism is said to be dominant-strategy incentive compatible, or truthful, if reporting truthfully is always the best response for each agent:
% %	\begin{gather*}
% %	u_i(x(t_i,t_{-i});t_i)\ge u_i(x(t'_i,t_{-i});t_i), \forall t_i\ge t'_i, \forall t_{-i}, \forall i.
% %	\end{gather*}
% 	\begin{gather*}	 
% 	 u_i(x(t_i,t_{-i}'),t)\ge  u_i(x(t_i',t_{-i}'),t), \forall t_i\ge t'_i, \forall t_{-i}\succeq t'_{-i}, \forall i,
% 	\end{gather*}
% 	where $\succeq$ means greater than element-wisely.
% \end{definition}
For ease of presentation, we say agent $i$ reports $\emptyset$ if he chooses not to participate (so we have $x_i(\emptyset,t_{-i}) = 0$ and $p_i(\emptyset,t_{-i}) = 0$). To encourage the agents to participate in the mechanism, the following property should be satisfied:
\begin{definition}[Individual Rationality (IR)]
	A mechanism is said to be individually rational, if no agent loses by participation when the other agents report truthfully:
	\begin{gather*}
	\textstyle
	u_i(x(t_i,t_{-i}),t) \ge u_i(x(\emptyset,t_{-i}),t), \forall t_i, t_{-i}, \forall i.
	\end{gather*}
%          \begin{gather*}
% u_i(x(t_i,t_{-i}'),t)\ge  u_i(x(\emptyset,t_{-i}'),t),\forall t_{-i}\succeq t'_{-i},\forall t, \forall i.
%          \end{gather*}
\end{definition}

The revenue and welfare of a mechanism are defined to be all the payments collected from the agents and all the valuations of the agents.
\begin{definition}
	The revenue and welfare of a mechanism $(x,p)$ are:
	\begin{gather*}
	\rev(x,p)=\sum_{i=1}^n p_i(t'),\\
	\wel(x,p)=\sum_{i=1}^n v_i(x,t).
	\end{gather*}
%	\begin{gather*}
%	\textstyle
%	\rev(x,p)=\sum_{i=1}^n p_i(t'), \
%	\wel(x,p)=\sum_{i=1}^n v_i(x,t).
%	\end{gather*}

	We say that a mechanism is {\em efficient} if 
	\begin{gather*}
	(x,p)=\argmax_{(x,p)}\wel(x,p),
	\end{gather*}
%	\begin{gather*}
%	\textstyle
%	 x= \mathop{\arg\max}_{x} \sum_i{v_i(x,t)}.
%	\end{gather*}
\end{definition}
A mechanism is \emph{weakly budget-balance} if it never loses money.
\begin{definition}[Weak Budget Balance]
	A mechanism is weakly budget-balance if:
	\begin{gather*}
	\textstyle
	\rev(x,p)\ge 0, \forall t.
	\end{gather*}
\end{definition}

\begin{definition}[Desirable Mechanism]
    We say a mechanism is desirable if it is IC, IR, efficient and weakly budget-balance.
\end{definition}

\subsection{Comparison with the Standard Interdependent Value Setting}
Although each agent's valuation depend on both the outcome of the mechanism and all agent's true types, our interdependent value with type-dependent action spaces setting, however, fundamentally different from standard interdependent value settings: %Recall in the competing companies example, the valuation (or market share) of a company depends on the models used by other companies to serve their customers. It means companies not joining in the multi-party learning also have influences on valuations of all others. The quality of model trained by the central platform depends on the updates that participants send, thus the set of available allocation outcomes is correlated with agents’ reported types in our setting. Models the mechanism allocates to agents are not necessarily used in the competition, since companies may have incentives to hide some of their data (to prevent others from getting better models) and use the model trained with all their data instead. Additionally, in our setting, each agent has a limited action space and cannot report a type higher than his actual type while both over-reporting and under-reporting are allowed in the standard setting. Our setting puts a stronger restriction on agents' behaviors, thus has a larger design space. 
%We emphasize the substantial differences on settings between the mechanism with type-dependent action spaces and the standard interdependent setting:
%Our setting exhibits several key differences from the standard auction setting:
\begin{itemize}
    \item In our setting, the type of each agent is the ``quality'' of his dataset, thus each agent cannot report a higher type than his true type. While in the standard interdependent value setting, an agent can possibly report any type.
    \item In our setting, the agents do not have the ``exit choice'' (not participating in the mechanism and getting 0 utility) as they do in the standard setting. This is due to the motivation of this paper: companies from the same industry are trying to improve their service quality, and they are always in the game regardless of their choices. A non-participating company may even have a negative utility if all other companies improved their services.
    %\item To capture the competition among the agents, the valuation function of each agent does not only depend on his type and allocation, but also possibly depends on other agents' types and allocations.
    \item To capture the cooperation among the agents, the item being sold, i.e., the learned model, also depends on all agents types. The best model learned by the multi-party learning platform will have high quality if all agents contribute high-quality datasets. However, the objects for allocation are usually fixed in standard mechanisms instead. 
\end{itemize}

\section{Quasi-monotone Externality Setting}\label{sec:linear}
In the interdependent value with type-dependent action spaces setting, each agent's utility may also depend on the models that other agents actually use. Such externalities lead to interesting and complicated interactions between the agents. For example, by contributing more data, one may improve the others' model quality, and end up harming his own market share.

In this section, we study the setting where agents have \emph{quasi-monotone} externalities. 
%This is a classic setting of externality that is well-studied in aution theory \cite{jehiel1996not,deng2011money,brocas2013optimal}. 
\begin{definition}[Quasi-Monotone Valuation]
    Let $q_i$ be the final selected model quality of the agent and $q_{-i}$ be the profile of model qualities of all the agents except $i$. A valuation function is quasi-monotone if it is in the form:
    $$v_i(q_i,q_{-i})=F_i(q_i)+\theta_i(q_{-i}),$$
    where $F_i$ is a monotone function and $\theta_i$ is an arbitrary function.
\end{definition}
\begin{example}
    Let's consider a special quasi-monotone valuation: the linear externality setting, where the valuation for each agent is defined as $v_i=\sum_j\alpha_{ij}{q_j}$ with $q_j$ being the model that agent $j$ uses. The externality coefficient $\alpha_{ij}$ means the influence of agent $j$ to agent $i$ and captures either the competitive or cooperative relations among agents. If the increase of agent $j$'s model quality imposes a negative effect on agent $i$'s utility (e.g. major opponents in the market), $\alpha_{ij}$ would be negative. Also $a_{ij}$ could be positive if agent $i$ and agent $j$ are collaborators. Additionally, $\alpha_{ii}$ should always be positive, naturally.
    %Here we give a mechanism if the model selection strategy of each agent is MAX and the externality is linear. This also gives an example to the violation of the Myerson-Satterthwaite theorem.
    
    %In the linear externality setting, the valuation for each agent is defined as $u_i=\sum_j\alpha_{ij}{Q_j(t_j)}-p_i$. Here the parameter $\alpha_{ij}$ denotes the influence of agent $j$ to agent $i$. This is a classic setting of externality. If the increase of agent $j$ impose a negative affect on agent $i$ (e.g. major opponents in market), the parameter $\alpha_{ij}$ would be negative. Also there could be positive parameters. For instance $\alpha_{ii}$ should all be positive, naturally.
    
    %There is another assumption, that all the agents could not use the model allocated to it and his own hidden data to train a better model. If an agent hide some model and he does not truthfully report all his data, he could only select one from the two models (the model allocated by the center and the one trained by his own hidden data).

    In the linear externality setting, the efficient allocation is straightforward. For each agent $i$, if $\sum_j{\alpha_{ji}}\ge 0$, the platform gives agent $i$ the training model with best possible quality. Otherwise, agent $i$ are not allocated any model if $\sum_j{\alpha_{ji}}<0$. %This allocation rule achieves the maximal social welfare, so it is efficient. After deciding the efficient allocation, we need to design the payment function in order to make all the agents join the multi-party learning and contribute all their data. 
\end{example}

We introduce a payment function called {\em maximal exploitation payment}, and show that the mechanism with efficient allocation and the maximal exploitation payment guarantees individually rationality, truthfulness, efficiency and revenue optimum.
\begin{definition}[Maximal Exploitation Payment (MEP)]
For a given allocation function $x$, suppose the agent $i$ reports a type $t_i'$ and the other agents report $t_{-i}'$, the maximal exploitation payment is to charge $i$ $$p_i(t_i',t_{-i}')=v_i(x(t_i',t_{-i}'),t_i', t_{-i}')-v_i(x(\emptyset,t_{-i}'),t_i',t_{-i}').$$ 
\end{definition}
%\com{Do we still need the line $q_i=\max\{x_i,Q(t_i')\}$?}
\begin{theorem}\label{thm:linear}
Under the quasi-monotone valuation setting, any mechanism with MEP is the mechanism with the maximal revenue among all IR mechanisms, and it is IC.
\end{theorem}

% \begin{proof}
% We will give a more general version of this theorem next.
% \end{proof}
%The proof is omitted since we will give a more general version of this theorem next. We will show that the mechanism in Theorem \ref{thm:linear} guarantees all the properties for a wider externality setting where the valuation function is {\em quasi-monotone}. 

%\com{Do we still need the line $q_i=\max\{x_i,Q(t_i')\}$?}
%\begin{theorem}
%Under the quasi-monotone valuation setting, any mechanism with the MEP rule guarantees IC, and is the mechanism that achieves the maximal revenue among all IR mechanisms. 
%\end{theorem}

\begin{proof}
Intuitively, the MEP rule charges agent $i$ the profit he gets from an model that the mechanism allocates to him. If the mechanism charges higher than the MEP, an agent would have negative utility after taking part in. The IR constraint would then be violated. So it's easy to see that the MEP is the maximal payment among all IR mechanisms. 

Then we prove that this payment rule also guarantees the IC condition. It suffices to show that if an agent hides some data, no matter which model he chooses to use, he would never get more utility than that of truthful reporting. We suppose that agent $i$'s type is $t_i'$ and he untruthfully reports $t_i'$.

Suppose that the agent $i$ truthfully reports the type $t_i'=t_i$, since the payment function is defined to charge this agent until he reaches the valuation when he does not take part in the mechanism, the utility of this honest agent would be 
\begin{gather*}
  u_i^{0}(t') = F_i(Q(t_i))+\theta_i(q_{-i}(\emptyset,t_{-i}')). 
\end{gather*}

If the agent does not report truthfully, we suppose that the agent reports $t_i'$ where $t_i' \le t_i$. According to the MEP, the payment function for agent $i$ would be 
\begin{gather*}
    p_i(t'_i,t_{-i}')=F_i(q_i(t_i',t_{-i}'))+\theta_i(q_{-i}(t_i',t_{-i}'))-F_i(Q(t_i'))-\theta_i(q_{-i}(\emptyset,t_{-i}')).
\end{gather*}
%Note that all the agents except $i$ report truthfully, so the qualities of their final selected model should be no worse than the best models they train with their own data. Therefore we assume the mechanism always gives them better than their best private models, i.e. $x_j \ge Q(t_j), \forall j \ne i$. So we have $q_{-i}(t_i',t_{-i}'=x_{-i}(t_i',t_{-i}'$.
%Note that it is without loss of generality to assume that the allocation $x_i(t'_i,t_{-i}') \ge Q(t'_i)$, $\forall t'_i,t_{-i}', \forall i$, since agents would never use the models allocated to them worse than those they train by themselves. Then we have $q_{-i}(t_i',t_{-i}')=x_{-i}(t_i',t_{-i}')$.
It can be seen that the mechanism would never give an agent a worse model than the model trained by its reported data, otherwise the agents would surely select their private data to train models. Hence it is without loss of generality to assume that the allocation $x_i(t'_i,t_{-i}') \ge Q(t'_i)$, $\forall t'_i,t_{-i}', \forall i$. Thus we have $q_{-i}(t_i',t_{-i}')=x_{-i}(t_i',t_{-i}')$.
We discuss the utility of agent $i$ by two cases of choosing models.

{\bf Case 1: the agent chooses the allocation $x_i$}. 
Since agent $i$ selects the allocated model, we have $q_i=x_i(t_i', t_{-i}')$. Then the utility of agent $i$ would be 
\begin{align*}
    u_i^{1}=&v_i(t_i', t_{-i}')-p_i(t_i', t_{-i}')\\
    =& F_i(x_i(t_i', t_{-i}')) + \theta_i(x_{-i}(t_i',t_{-i}'))-p_i(t_i', t_{-i}')\\
    =& F_i(x_i(t_i', t_{-i}')) + \theta_i(x_{-i}(t_i',t_{-i}'))+F_i(Q(t_i'))
    \\
    &+\theta_i(x_{-i}(\emptyset,t_{-i}'))-F_i(x_i(t_i',t_{-i}'))-\theta_i(x_{-i}(t_i',t_{-i}'))\\
    =&F_i(Q(t_i'))+\theta_i(x_{-i}(\emptyset,t_{-i}')).
\end{align*}
Because both $F_i$ and $Q$ are monotone increasing functions and $t_i\ge t_i'$, we have $u_i^{1} \le F_i(Q(t_i))+\theta_i(x_{-i}(\emptyset,t_{-i}')) = u_i^{0}$.

{\bf Case2: the agent chooses $Q(t_i)$}.
Since agent $i$ selects the model trained by his private data, we have $q_i = Q(t_i)$. The final utility of agent $i$ would be 
\begin{align*}
    u_i^{2}=&v_i(t_i', t_{-i}')-p_i(t_i', t_{-i}')\\
    =& F_i(Q(t_i)) + \theta_i(x_{-i}(t_i',t_{-i}'))-p_i(t_i', t_{-i}')\\
    =& F_i(Q(t_i)) + \theta_i(x_{-i}(t_i',t_{-i}'))+F_i(Q(t_i'))
    \\
    &+\theta_i(x_{-i}(\emptyset,t_{-i}'))-F_i(x_i(t_i',t_{-i}'))-\theta_i(x_{-i}(t_i',t_{-i}'))\\
    =&F_i(Q(t_i)) + F_i(Q(t_i')) + \theta_i(x_{-i}(\emptyset,t_{-i}')) - F_i(x_i(t_i',t_{-i}')). 
\end{align*}
Subtract the original utility from the both sides, then we have 
\begin{align*}
    u_i^2 - u_i^0 =& F_i(Q(t_i)) + F_i(Q(t_i')) + \theta_i(x_{-i}(\emptyset,t_{-i}'))\\
    &- F_i(x_i(t_i',t_{-i}')) - F_i(Q(t_i))-\theta_i(x_{-i}(\emptyset,t_{-i}'))\\
    =& F_i(Q(t_i')) - F_i(x_i(t_i',t_{-i}')).
\end{align*}
Because $x_i(t'_i,t_{-i}') \ge Q(t'_i)$, $\forall t'_i,t_{-i}', \forall i$ and because $F_i$ is a monotonically increasing function, we can get $u_i^2 - u_i^0 \le 0$. Therefore $\max\{u_i^1, u_i^2\}$ $\le u_i^0$, lying would not bring more benefits to any agent, and the mechanism is IC.
\end{proof}
% Assume that $x(\cdot)$ and $p(\cdot)$ are differentiable, the IC condition is equivalent to 
% \begin{gather*}
%     p_i(t_i',t_{-i}')-p_i(t'_i,t_{-i}') \leq \sum^n_{j \neq i}\alpha_{ij}\int_{t'_i}^{t_i'}\frac{\partial x_j(s,t_{-i}')}{\partial s}\mathrm{d}s + F
% \end{gather*}
% where $F=\min\{x_i(t_i',t_{-i}')-Q(t_i'),x_i(t_i',t_{-i}')-x_i(t'_i,t_{-i}')\}$. 
% \begin{theorem}
% Under the linear externality utility setting and MAX hiding strategy (take a better model from the allocated one and the model trained by an agent's all data). a mechanism is IC if and only if $\forall i, \forall t_{-i}'$,
% \begin{align*}
%     \frac{\partial p_i(s,t_{-i}')}{\partial s} \leq& \sum^n_{j=1}\alpha_{ij}\frac{\partial x_i(s,t_{-i}')}{\partial s}\\
%     p_i(s,t_{-i}') \leq& \sum^n_{j =1}\alpha_{ij}x_j(s, t_{-i}') - Q(s) + \min_{s' \leq s}(p(s,t_{-i}')-\sum^n_{j \ne i}\alpha_{ij}x_j(s', t_{-i}'))
% \end{align*}
% \end{theorem}

\begin{corollary}\label{col:linear}
	Any efficient allocation mechanism with MEP under the linear externality setting with all the linear coefficients $\alpha_{ji}\ge 0$ should be IR, IC, weakly budget-balance and efficient.
\end{corollary}
\begin{proof}
In Theorem \ref{thm:linear} we know that the MEP mechanism is IR and IC. Since the linear coefficients are all positive and the externality setting is linear, any efficient mechanism would allocate the best model to all the agents. Since each agent gets a model with no less quality than his reported one and the payment is equal to the value difference between the case an agent truthfully report and the case he exit the mechanism. The agent's value is always larger than the value when he exits the mechanism. Then the payment is always positive and the mechanism should satisfy all of the four properties.
\end{proof}
In the standard mechanism design setting, the Myerson-Satterthwaite Theorem \cite{myerson1983efficient} is a well-known classic result, which says that no mechanism is simultaneously IC, IR, efficient and weakly budget-balance.
The above Corollary ~\ref{col:linear} shows that in our setting, the  Myerson-Satterthwaite Theorem fails to hold.

\section{General Externality Setting}
%In this section, we investigate the mechanism in the general externality setting. In this setting, the valuation of agent $i$ can be any function of the final model qualities of all agents. We characterize the existence of the desirable mechanism by an algorithm. Then we analyze some properties of truthful mechanisms.
In this section, we consider the general externality setting where the valuation of agent $i$ can be any function of allocation outcome and types of all agents. The limitation on the reporting space and the type-imposed value functions make the IC and IR mechanisms hard to characterize. It is possible that given a allocation rule, there exist several mechanisms with different payments that satisfy both IC and IR constraints. To understand what makes a mechanism IC and IR, we analyze some properties of truthful mechanisms in this section.

For ease of presentation, we assume that the functions $v(\cdot)$, $x(\cdot)$ and $p(\cdot)$ are differentiable. 

\begin{theorem}[Necessary Condition]\label{thm:necessary}
    If a mechanism $(x,p)$ is both IR and IC, for all possible valuation functions satisfying Assumption \ref{asm:value_monotone}, then the payment function satisfies $\forall t_i \ge t'_i, \forall t_i,\forall t_{-i},\forall i$,
    \begin{gather}
	p_i(0, t_{-i})\le v_i(x(0, t_{-i}),0,t_{-i})-v_i(x(\emptyset, t_{-i}),0,t_{-i}),\label{eq:ir}\\
	p_i(t_i,t_{-i})-p_i(t'_i,t_{-i})
	\le \int_{t'_i}^{t_i}\left.\frac{\partial v_i(x(s', t_{-i}),s,t_{-i})}{\partial s'}\right|_{s=s'} \,\mathrm{d}s', \label{eq:ic_payment}
	\end{gather}
	where we view $v_i(x(t'_i, t_{-i}),t_i,t_{-i})$ as a function of $t_i$, $t'_i$ and $t_{-i}$ for simplicity and
	\begin{align*}
	\frac{\partial v_i(x(t'_i, t_{-i}),t_i,t_{-i})}{\partial t'_i}
	=\sum_{i=1}^n\frac{\partial v_i(x(t'_i, t_{-i}),t_i,t_{-i})}{\partial x_j(t'_i,t_{-i})}\frac{\partial x_j(t'_i,t_{-i})}{\partial t'_i}.
	\end{align*}
%	\begin{align*}
%	    &\frac{\partial v_i(x(t'_i, t'_{-i}),t_i,t_{-i})}{\partial t'_i}\\
%	    =&\sum_{i=1}^n\frac{\partial v_i(x(t'_i, t'_{-i}),t_i,t_{-i})}{\partial x_j(t'_i,t'_{-i})}\frac{\partial x_j(t'_i,t'_{-i})}{\partial t'_i}
%	\end{align*}
\end{theorem}
\begin{proof}
	We first prove that Equation (\ref{eq:ir}) holds. Observe that
	\begin{align}
	&u_i(x(t_i,t'_{-i}),t_i,t_{-i})-u_i(x(t'_i,t'_{-i}),t'_i,t_{-i})\nonumber\\
	=&[v_i(x(t_i,t'_{-i}),t_i,t_{-i})-p_i(t_i,t'_{-i})]
	-[v_i(x(t'_i,t'_{-i}),t'_i,t_{-i})-p_i(t'_i,t'_{-i})]\nonumber\\
	\ge&[v_i(x(t_i,t'_{-i}),t_i,t_{-i})-p_i(t_i,t'_{-i})]
	-[v_i(x(t'_i,t'_{-i}),t_i,t_{-i})-p_i(t'_i,t'_{-i})]\nonumber\\
	=&u_i(x(t_i,t'_{-i}),t_i,t_{-i})-u_i(x(t'_i,t'_{-i}),t_i,t_{-i})\nonumber\\
	\ge &0, \label{eq:ic_to_ir}
	\end{align}
	where the first inequality is because of Assumption \ref{asm:value_monotone}, and the last inequality is because of the DSIC property.
	
	Let $t'_i=0$ in Equation (\ref{eq:ic_to_ir}). We have
	\begin{gather*}
	u_i(x(t_i,t'_{-i}),t_i,t_{-i})\ge u_i(x(0,t'_{-i}),0,t_{-i}).
	\end{gather*}
	The IR property further requires that $ u_i(x(0,t'_{-i}),0,t_{-i})\ge u_i(x(\emptyset, t'_{-i}),0,t_{-i})$, which Equation (\ref{eq:ir}) follows.
	
	To show Equation (\ref{eq:ic_payment}) must hold, we rewrite Equation (\ref{eq:ic_to_ir}):
	\begin{align}
	p_i(t_i,t'_{-i})-p_i(t'_i,t'_{-i})
	\le& v_i(x(t_i,t'_{-i}),t_i,t_{-i})-v_i(x(t'_i,t'_{-i}),t'_i,t_{-i})\nonumber\\
	=&\int_{t'_i}^{t_i}\frac{\mathrm{d}v_i(x(s',t'_{-i}),s(s'),t_{-i})}{\mathrm{d}s'}\,\mathrm{d}s'.\label{eq:ic_payment_proof}
	\end{align}
	
	Fixing $t_{-i}$ and $t'_{-i}$, the total derivative of the valuation function $v_i(x(s', t'_{-i}),s,t_{-i})$ is:
	%		\begin{gather*}
	%		\mathrm{d}v_i(x(s',t'_{-i}),s,t_{-i})=\frac{\partial v_i(x(s', t'_{-i}),s,t_{-i})}{\partial s'}\,\mathrm{d}s'+\frac{\partial v_i(x(s', t'_{-i}),s,t_{-i})}{\partial s}\,\mathrm{d}s
	%		\end{gather*}
	\begin{align*}
%	\textstyle
	\mathrm{d}v_i(x(s',t'_{-i}),s,t_{-i})
	=\frac{\partial v_i(x(s', t'_{-i}),s,t_{-i})}{\partial s'}\,\mathrm{d}s'+\frac{\partial v_i(x(s', t'_{-i}),s,t_{-i})}{\partial s}\,\mathrm{d}s.
	\end{align*}
	View $s$ as a function of $s'$ and let $s(s')=s'$:
	\begin{align*}
%	\textstyle
	\frac{\mathrm{d}v_i(x(s',t'_{-i}),s(s'),t_{-i})}{\mathrm{d}s'}
	=\left.\frac{\partial v_i(x(s', t'_{-i}),s,t_{-i})}{\partial s'}\right|_{s=s'}
	+\frac{\partial v_i(x(s', t'_{-i}),s(s'),t_{-i})}{\partial s(s')}\frac{\mathrm{d}s(s')}{\mathrm{d}s'}.
	\end{align*}
	Plug into Equation (\ref{eq:ic_payment_proof}), and we obtain:
		\begin{align*}
		p_i(t_i,t'_{-i})-p_i(t'_i,t'_{-i})
		\le\int_{t'_i}^{t_i}\left.\frac{\partial v_i(x(s', t'_{-i}),s,t_{-i})}{\partial s'}\right|_{s=s'}
		+\int_{t'_i}^{t_i}\frac{\partial v_i(x(s', t'_{-i}),s(s'),t_{-i})}{\partial s(s')}\,\mathrm{d}s'.
		\end{align*}
		Since the above inequality holds for any valuation function with $v_i(x,t_i,t_{-i})\ge v_i(x,t'_i,t_{-i}), \forall x, \forall t_{-i},\forall t_i\ge t'_i$, we have:
		\begin{align*}
		p_i(t_i,t'_{-i})-p_i(t'_i,t'_{-i})
		\le \int_{t'_i}^{t_i}\left.\frac{\partial v_i(x(s', t'_{-i}),s,t_{-i})}{\partial s'}\right|_{s=s'}\,\mathrm{d}s'.
		\end{align*}
	
	%		Finally, we show that Equation (\ref{eq:ic_allocation}) must hold. Suppose not. Then there exists $t_i> t'_i$ satisfying:
	%		\begin{gather*}
	%		\left.\frac{\partial v_i(x(t'_i, t_{-i});s)}{\partial t'_i}\right|_{s=t'_i}> \frac{\partial v_i(x(t'_i, t_{-i});t_i)}{\partial t'_i}
	%		\end{gather*}
	%		There must be a small number $\varepsilon>0$, such that:
	%		\begin{align}
	%		&v_i(x(t'_i, t_{-i});t'_i)-v_i(x(t'_i-\varepsilon, t_{-i});t'_i)\nonumber\\
	%		>&v_i(x(t'_i, t_{-i});t_i)-v_i(x(t'_i-\varepsilon, t_{-i});t_i)\label{eq:contradiction}
	%		\end{align}
	%		On the other hand, the DIC property requires $u_i(x(t_i,t_{-i});t_i)\ge u_i(x(t'_i,t_{-i});t_i)$, or equivalently,
	%		\begin{gather*}
	%		p_i(t_i,t_{-i})-p_i(t'_i,t_{-i})\le v_i(x(t_i,t_{-i});t_i)-v_i(x(t'_i,t_{-i});t_i)
	%		\end{gather*}
	%		Similarly,
	%		\begin{gather*}
	%		p_i(t'_i,t_{-i})-p_i(t'_i-\varepsilon,t_{-i})\le v_i(x(t'_i,t_{-i});t'_i)-v_i(x(t'_i-\varepsilon,t_{-i});t'_i)
	%		\end{gather*}
	%		Combining the above two inequalities with Equation (\ref{eq:contradiction}), we get:
	%		\begin{gather*}
	%		p_i(t_i,t_{-i})-p_i(t'_i-\varepsilon,t_{-i})\le v_i(x(t'_i,t_{-i});t'_i)-v_i(x(t'_i-\varepsilon,t_{-i});t'_i)
	%		\end{gather*}
\end{proof}
Theorem \ref{thm:necessary} describes what the payment $p$ is like in all IC and IR mechanisms. In fact, the conditions in Theorem \ref{thm:necessary} are also crucial in making a mechanism truthful. However, to ensure IC and IR property, we still need to restrict the allocation function $x$.
\begin{theorem}[Sufficient Condition]\label{thm:sufficient}
	A mechanism $(x,p)$ satisfies both IR and IC, for all possible valuation functions satisfying Assumption \ref{asm:value_monotone},  if for each agent $i$, for all $t_i\ge t'_i$, and all $t_{-i}$, Equations (\ref{eq:ir}) and the following two hold
	\begin{align}
	&t'_i=\argmin_{t_i:t_i>t'_i}\frac{\partial v_i(x(t'_i,t_{-i}),t_i,t_{-i})}{\partial t'_i}\label{eq:ic_allocation}\\
		p_i(t_i,t_{-i})-p_i(t'_i,t_{-i})
		\le&\int_{t'_i}^{t_i}\left.\frac{\partial v_i(x(s', t_{-i}),s,t_{-i})}{\partial s'}\right|_{s=s'}\,\mathrm{d}s' 
		- \int_{t'_i}^{t_i}\frac{\partial v_i(x(\emptyset,t_{-i}),s,t_{-i})}{\partial s}\mathrm{d}s\label{eq:ir_payment}.
	\end{align}
\end{theorem}
\begin{proof}	
	Equation (\ref{eq:ic_allocation}) indicates that the function $\frac{\partial v_i(x(t'_i, t'_{-i}),t_i,t_{-i})}{\partial t'_i}$ is minimized at $t'_i$:
	\begin{gather}
	\left.\frac{\partial v_i(x(t'_i, t'_{-i}),s,t_{-i})}{\partial t'_i}\right|_{s=t'_i}\le \frac{\partial v_i(x(t'_i, t'_{-i}),t_i,t_{-i})}{\partial t'_i} \label{eq:monotone}.
	\end{gather}
	Therefore, we have
	\begin{align}
%	\textstyle
	&u_i(x(t_i,t'_{-i}),t_i,t_{-i})-u_i(x(t'_i,t'_{-i}),t_i,t_{-i})\nonumber\\
	%	=&v_i(x(t_i,t'_{-i}),t_i,t_{-i})-v_i(x(t'_i,t'_{-i}),t_i,t_{-i})\nonumber\\
	%	&-p_i(t_i,t'_{-i})+p_i(t'_i,t'_{-i})\nonumber\\
	=&\int_{t'_i}^{t_i}\frac{\partial v_i(x(s',t'_{-i}),t_i,t_{-i})}{\partial s'}\,\mathrm{d}s'-p_i(t_i,t'_{-i})+p_i(t'_i,t'_{-i})\nonumber\\
	\ge&\int_{t'_i}^{t_i}\left.\frac{\partial v_i(x(s',t'_{-i}),s,t_{-i})}{\partial s'}\right|_{s=s'}\,\mathrm{d}s'\nonumber\\
	&-p_i(t_i,t'_{-i})+p_i(t'_i,t'_{-i})\nonumber\\
	\ge& \int_{t'_i}^{t_i}\frac{\partial v_i(x(\emptyset,t'_{-i}),s,t_{-i})}{\partial s}\mathrm{d}s \label{eq:ic_proof},
	%	\ge& \int_{t'_i}^{t_i}\frac{\partial v_i(x(t_{-i}),s,t_{-i})}{\partial s}\mathrm{d}s \nonumber\\
	%	\ge &0,\label{eq:ic_proof}
	\end{align}
	where the two inequalities are due to Equation (\ref{eq:monotone}) and (\ref{eq:ir_payment}), respectively. Since $v_i(x,t_i,t_{-i})$$\ge$$v_i(x,t'_i,t_{-i})$, $\forall x,$$\forall t_{-i},$$\forall t_i\ge t'_i$ indicates $\frac{\partial v_i(x(\emptyset,t'_{-i}),s,t_{-i})}{\partial s}\ge 0$, the above inequality shows that the mechanism guarantees the DSIC property.
	
	To prove that the mechanism is IR, we first observe that 
	%     	complete version
	%	\begin{align*}
	%	&[u_i(x(t_i,t'_{-i}),t_i,t_{-i})- v_i(x(\emptyset ,t'_{-i}),t_i,t_{-i})]\\
	%	&-[u_i(x(t'_i,t'_{-i}),t'_i,t_{-i}) -v_i(\emptyset,x(t_{-i}),t'_i,t_{-i})]\\
	%	=&u_i(x(t_i,t'_{-i}),t_i,t_{-i}) - u_i(x(t'_i,t'_{-i}),t'_i,t_{-i})  \\ &-(v_i(x(\emptyset,t'_{-i}),t_i,t_{-i}) - v_i(\emptyset,x(t'_{-i}),t'_i,t_{-i}))\\
	%	=&[v_i(x(t_i,t'_{-i}),t_i,t_{-i})-p_i(t_i,t'_{-i})]\nonumber\\
	%	&-[v_i(x(t'_i,t'_{-i}),t'_i,t_{-i})-p_i(t'_i,t'_{-i})]\nonumber\\
	%	&- \int_{t'_i}^{t_i}\frac{\partial v_i(x(\emptyset,t'_{-i}),s,t_{-i})}{\partial s}\mathrm{d}s\\
	%	\ge &[v_i(x(t_i,t'_{-i}),t_i,t_{-i})-p_i(t_i,t'_{-i})]\\
	%	&-[v_i(x(t'_i,t'_{-i}),t_i,t_{-i})-p_i(t'_i,t'_{-i})] \\
	%	&-\int_{t'_i}^{t_i}\frac{\partial v_i(x(\emptyset,t'_{-i}),s,t_{-i})}{\partial s}\mathrm{d}s \\
	%	=&u_i(x(t_i,t'_{-i}),t_i,t_{-i}) - u_i(x(t'_i,t'_{-i}),t_i,t_{-i}) \nonumber\\
	%	&- \int_{t'_i}^{t_i}\frac{\partial v_i(x(\emptyset,t'_{-i}),s,t_{-i})}{\partial s}\mathrm{d}s \\
	%	\ge&0,
	%	\end{align*}
	\begin{align*}
	&[u_i(x(t_i,t'_{-i}),t_i,t_{-i})- v_i(x(\emptyset ,t'_{-i}),t_i,t_{-i})]
	-[u_i(x(t'_i,t'_{-i}),t'_i,t_{-i}) -v_i(\emptyset,x(t_{-i}),t'_i,t_{-i})]\\
	%	=&u_i(x(t_i,t'_{-i}),t_i,t_{-i}) - u_i(x(t'_i,t'_{-i}),t'_i,t_{-i})  \\ &-(v_i(x(\emptyset,t'_{-i}),t_i,t_{-i}) - v_i(\emptyset,x(t'_{-i}),t'_i,t_{-i}))\\
	=&u_i(x(t_i,t'_{-i}),t_i,t_{-i}) - u_i(x(t'_i,t'_{-i}),t'_i,t_{-i})
	- \int_{t'_i}^{t_i}\frac{\partial v_i(x(\emptyset,t'_{-i}),s,t_{-i})}{\partial s}\mathrm{d}s\\
	\ge &u_i(x(t_i,t'_{-i}),t_i,t_{-i}) - u_i(x(t'_i,t'_{-i}),t_i,t_{-i})
	-\int_{t'_i}^{t_i}\frac{\partial v_i(x(\emptyset,t'_{-i}),s,t_{-i})}{\partial s}\mathrm{d}s \\
	\ge&0,
	\end{align*}
	where the two inequalities are Assumption \ref{asm:value_monotone} and Equation (\ref{eq:ic_proof}). Letting $t'_i=0$ using Equation (\ref{eq:ic_payment}), we get:
	\begin{align*}
	& u_i(x(t_i,t'_{-i}),t_i,t_{-i}) - v_i(x(\emptyset,t'_{-i}),t_i,t_{-i})\\
	\ge& u_i(x(0,t'_{-i}),0,t_{-i})- v_i(x(\emptyset,t'_{-i}),0,t_{-i})\\
	=& v_i(x(0,t'_{-i}),0,t_{-i})-p_i(0,t'_{-i})- v_i(x(\emptyset,t'_{-i}),0,t_{-i})\\
	\ge & 0.
	\end{align*}	
\end{proof}
%We omit the proofs of the theorems in this section due to the lack of space.

\section{Market Growth Rate}

In this section, we will analyze a factor, the market growth rate, for the existence of the desirable mechanism. Expanding the market size would reduce competition among the agents, meaning that the damage to an agent's existing market caused by joining the mechanism is more likely to be covered by the market growth. Thus our intuition is that if the market grows quickly, a desirable mechanism is more likely to exist. 

As mentioned above, each agent's valuation is the profit made from the market, so formally we define the market size to be the sum of the valuations of all the agents. Let $M(q)$ be the agents’ total valuations where $q=(q_1,q_2,\dots,q_n)$ is the set of actual model qualities they use. We have:
\begin{gather*}
%\textstyle
    M(q) = \sum_{i=1}^{n}v_i(x,t).
\end{gather*}
In general, the multi-party learning process improves all agents' models. So we do not consider the case where the market shrinks due to the agents' participation, and assume that the market is growing.
\begin{assumption}[Growing Market]
$q \succeq q'$ implies $M(q)\ge M(q')$.
\end{assumption}

A special case of the growing market is the non-competitive market where agent's values are not affected by others' model qualities, formally:
\begin{definition}[Non-competitive Market]
A market is non-com-petitive if and only if $\dfrac{\partial v_i(q)}{\partial q_j}\ge 0,\forall i, j$. 
\end{definition}
\begin{theorem}\label{nc-m}
In a non-competitive market, there always exists a desirable mechanism, that gives the best possible model to each agent and charges nothing.
\end{theorem}
\begin{proof}
Suppose that the platform uses the mechanism mentioned in the theorem. Then for each agent, contributing with more data increases all participants' model qualities. By definition, in a non-competitive market, improving others' models does not decrease one's profit. Therefore, the optimal strategy for each participant is to contribute with all his valid data, making the mechanism truthful. Also because of the definition, entering the platform always weakly increases one's model quality. Thus the mechanism is IR. With the IC and IR properties, it is easy to see that the mechanism is also efficient and weakly budget-balance.
\end{proof}
However, when the competition exists in a growing market, it is not easy to determine whether a desirable mechanism exists. Since when the market is growing, the efficient mechanism both redistributes existing markets and enlarges the market size by giving the best learned model. We will give the empirical analysis of the influence of competition and the market growth rate on the existence of desirable mechanisms.

\section{Experiments}\label{sec:exp}
We design experiments to demonstrate the performance of our mechanism for practical use. We first show the mechanism with maximal exploitation payments (MEP) can guarantee good quality of trained models and high revenues under the linear externality cases. Then we conduct simulations to show the relation of the market growth of competitive market to the existence of desirable mechanisms.
\subsection{The MEP Mechanism}
We consider the valuation with linear externalities setting where $\alpha_{ij}$'s (defined in Example \ref{linear_example}) are generated uniformly in $[-1,1]$.  Each agent's type is drawn uniformly from $[0,1]$ independently and the $Q(t)$ is $\frac{1-e^{-t}}{1+e^{-t}}$. The performance of a mechanism is measured by the platform's revenue and its best quality of trained model under the mechanism. All the values of each instance are averaged over 50 samples. We both show the performance changes as the number of agents increases and as the agents' type changes. 

When the number of agents becomes larger, the platform can obtain more revenues and train better models (see Figure \ref{fig:number}). Particularly, the model quality is close to be optimal when the number of agents over $12$. An interesting phenomenon is that the revenue may surpass the social welfare. This is because the average external effect of other agents on one agent $i$ tends to be negative when agent $i$ does not join in the mechanism, thus the second term in the MEP payment is averagely negative and revenue is larger than the welfare.  

To see the influence of type on performance, we fix one agent's type to be 1 and set the other agent's type from 0 to 2. It can be seen in Figure \ref{fig:type} that the welfare and opponent agent's utility (\texttt{uti\_2}) increase as the opponent's type increases but the platform's revenue and the utility of the static agent (\texttt{uti\_1}) are almost not affected by the type. So we draw the conclusion that the most efficient way for the platform to earn more revenue is to attract more small companies to join the mechanism, since in the figure \ref{fig:number} the revenue obviously increases as the number of agents increases. 
\begin{figure}
	\centering
	\subfigure[Welfare \& revenue]{\includegraphics[width=0.46\textwidth]{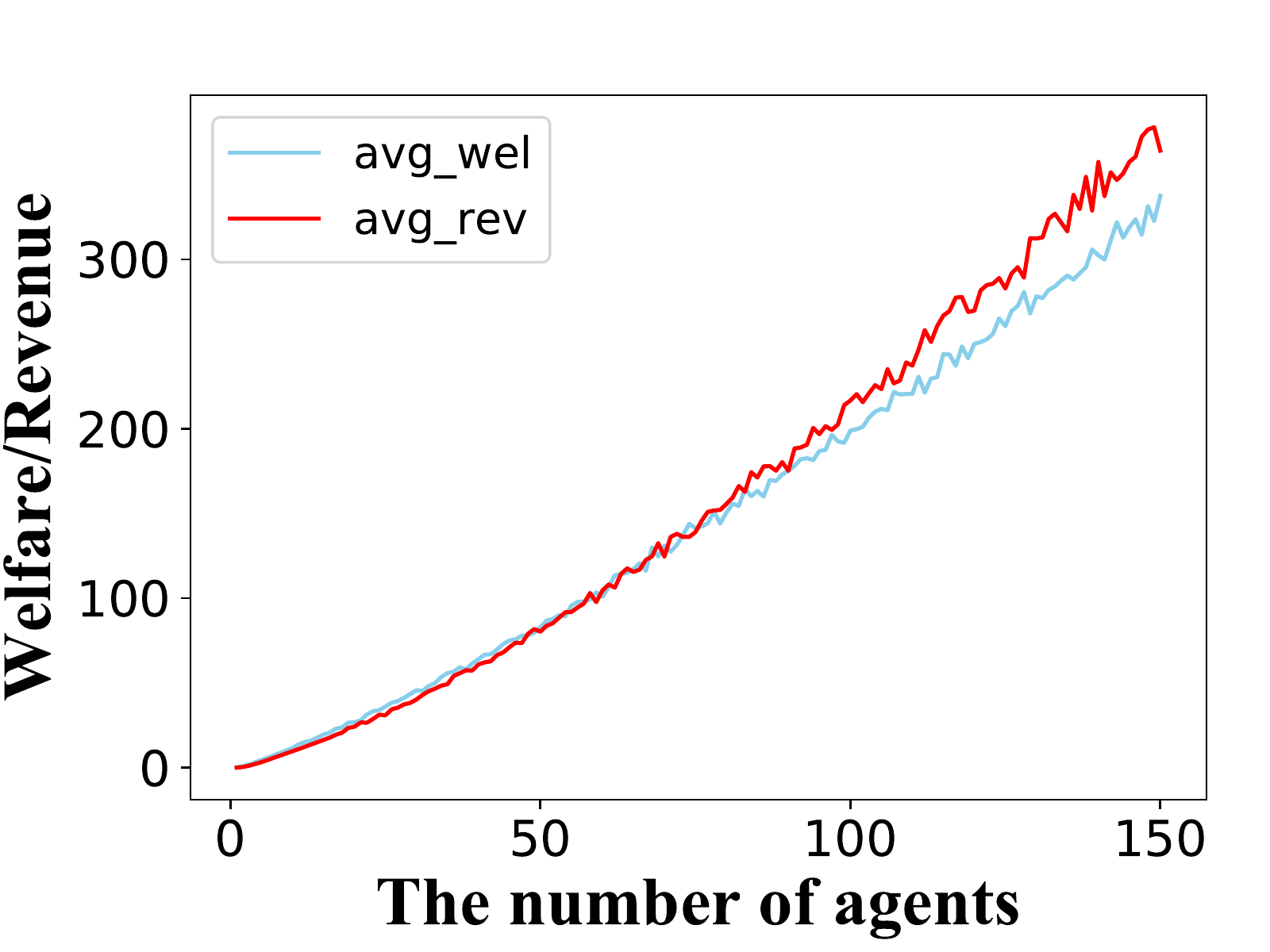}\label{sfig:value_number}}%
	\hfill %
	\subfigure[Best quality of trained model]{\includegraphics[width=0.48\textwidth]{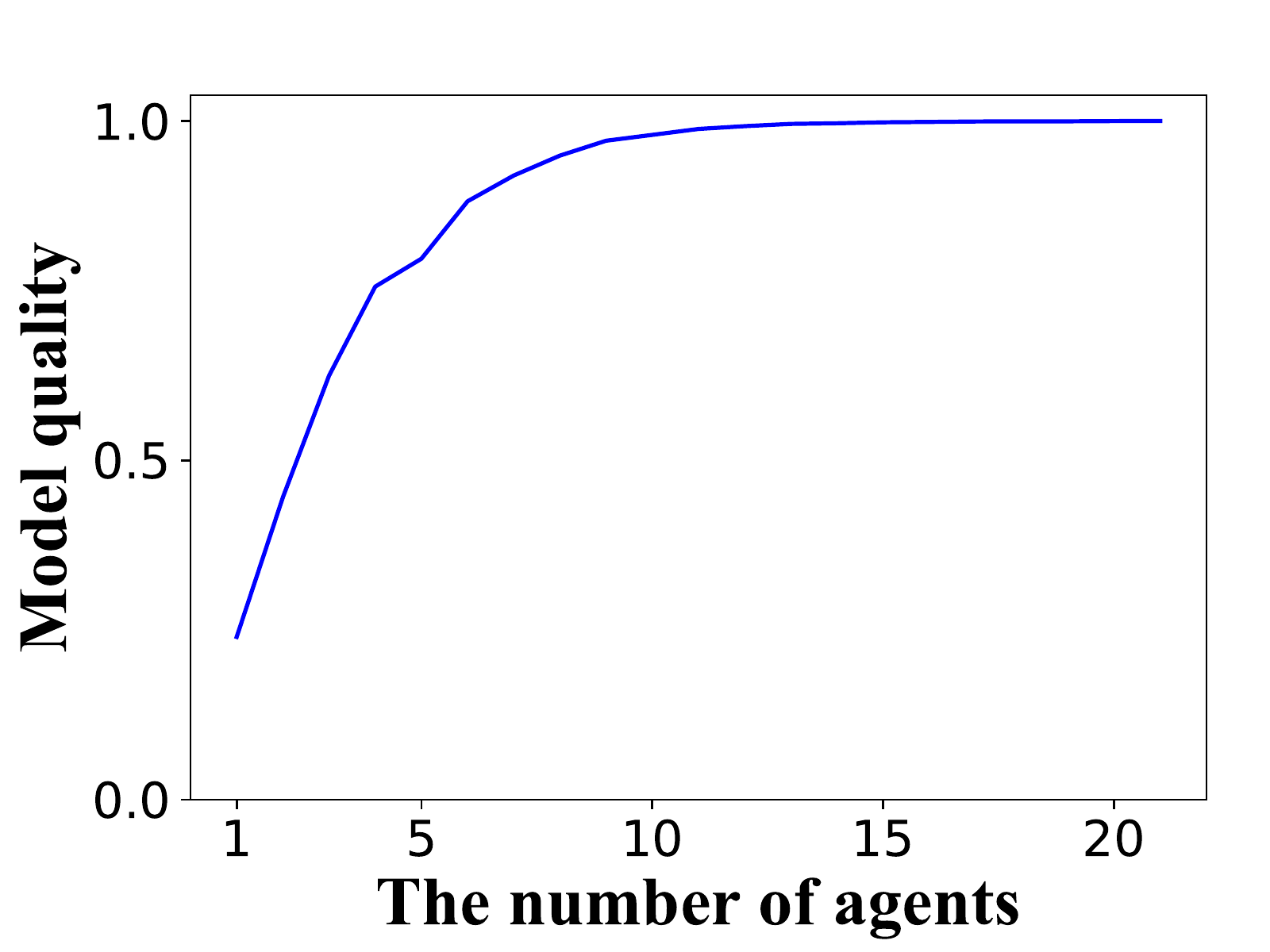}\label{sfig:quality_number}}
	\caption{Performance of MEP under different numbers of agents}
	\label{fig:number}
	%\vspace{-1.7\baselineskip}
\end{figure}
\begin{figure}
	\centering
	\includegraphics[width=0.52\textwidth]{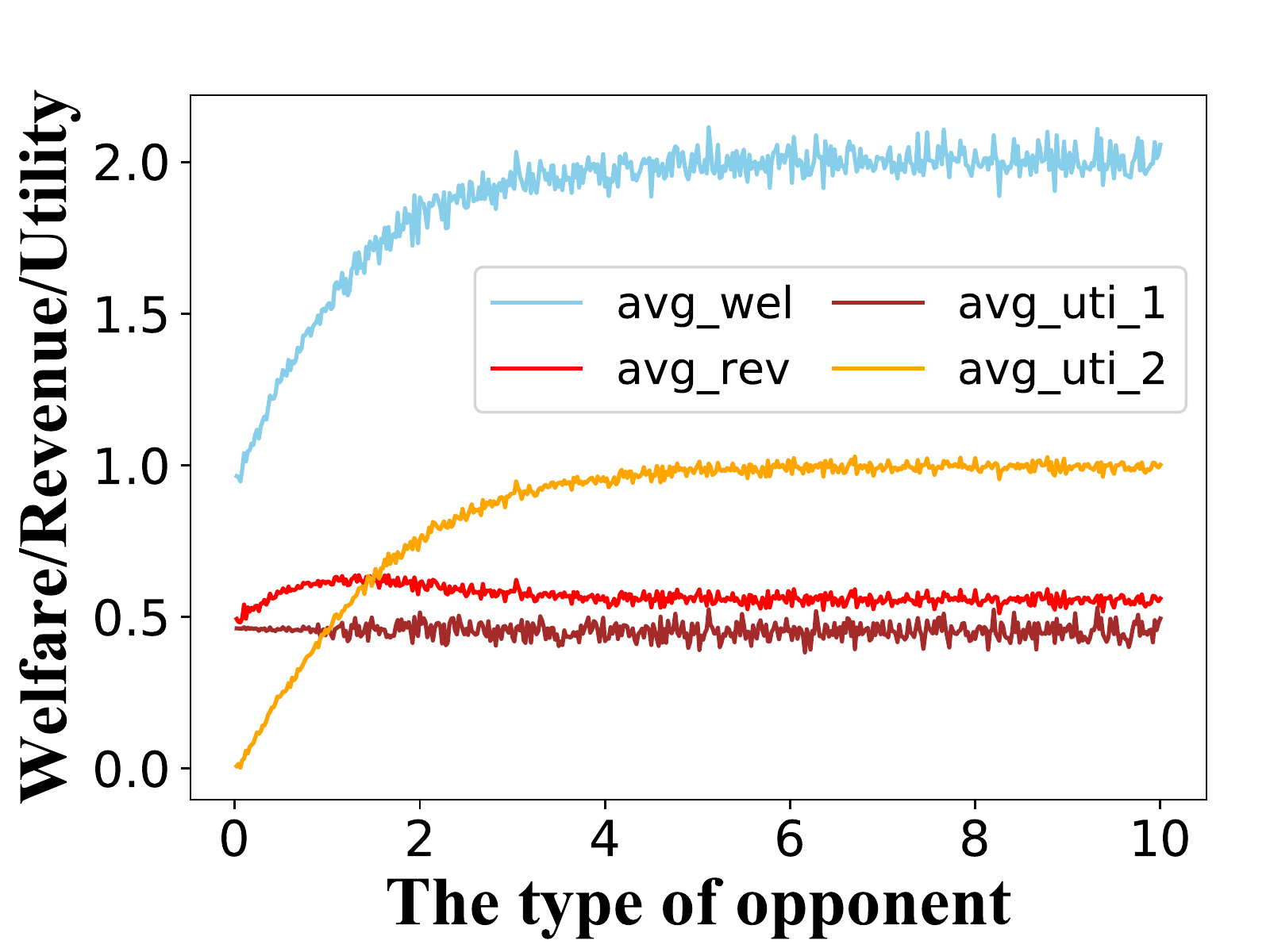}
	\caption{Performance of MEP under different types}
	\label{fig:type}
\end{figure}
\subsection{Existence for Desirable Mechanisms}
We assume all the agents' type lies in $[0, D]$, and the type space can be discretized into intervals of length $\epsilon$, which is also the minimal size of the data. Thus each agent’s type is a multiple of $\epsilon$. The data disparity is defined as the ratio of the largest possible data size to the smallest possible data size,  denoted as, $D/\epsilon$. We measure the condition for existence of desirable mechanisms by the maximal data disparity when the market growth rate is given. 

To describe the market growth, we use the following form of valuation function and model quality function:
\begin{gather*}
%\textstyle
Q(t)=t \text{ and }
v_i(q)=\left(\sum_{j = 1}^{n}Q(q_j)\right)^\alpha\cdot Q(q_{i}), \forall i,
\end{gather*}
where the parameter $\alpha$ indicates the market growth rate. As we can see, when $\alpha < -1$, the market is not a growing market. When $\alpha \ge0$, the market becomes non-competitive, therefore by Theorem \ref{nc-m}, a desirable mechanism trivially exists. Thus we consider the competitive growing market case where $-1\le\alpha <0$.

We provide an algorithm to find the desirable mechanisms. Since the utility function is general, all the points in the action space would influence the properties and existence of the mechanism, thus the input space of the algorithm is exponential. The algorithm can be seen in Appendix \ref{algo}. We enumerate the value of $\alpha$ and run the algorithm to figure out the boundary of $D/\epsilon$ under different $\alpha$ in a market with 2 agents. It turns out that when $alpha$ is near $-0.67$, the boundary of disparity grows very fast. When $\alpha$ is close to $ -0.66$, the disparity boundary is over 10000 and is beyond our computing capability, so we just enumerate $\alpha$ from $-1$ to $-0.668$ with step length $0.002$.
\begin{figure}[t!]
	\centering
	\includegraphics[width=0.52\textwidth]{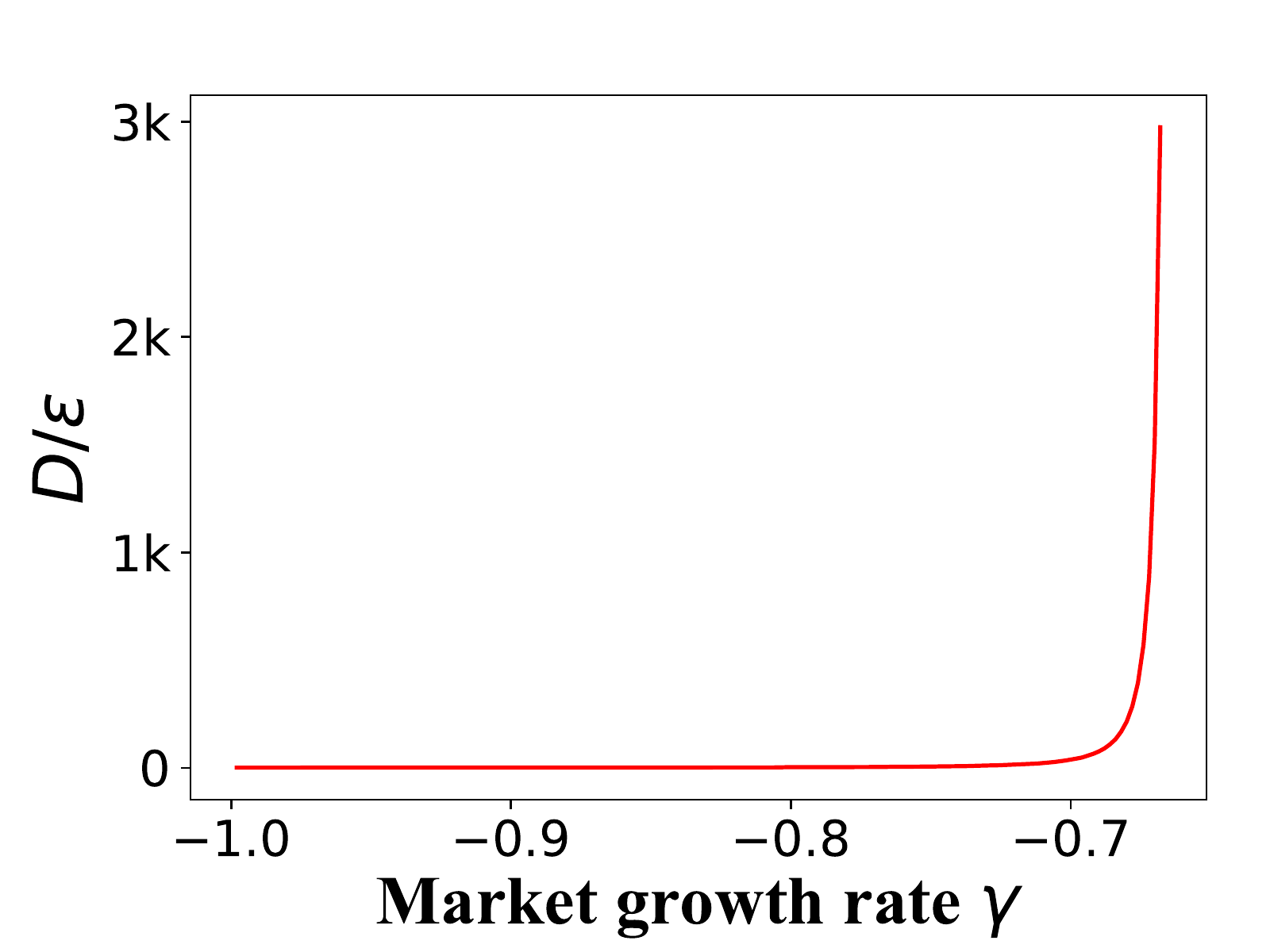}
	%\short{\vspace{-0.7\baselineskip}}{}
	\caption{The boundary of data disparity for existence of desirable mechanisms}
	\label{fig:exis}
\end{figure}

Figure \ref{fig:exis} demonstrates the maximal data disparity given the existence of a desirable mechanism under different market growth rates. 
The y-axis is the data disparity $D/\epsilon$ and the x-axis shows the market growth rate $\alpha$. For every fixed $\alpha$, there does not exist any desirable mechanism when the biggest size of the agent's dataset is larger than $D$ if the smallest size of dataset is $\epsilon$. It can be seen an obvious trend that when $\alpha$ becomes larger, the constraint on data size disparity would become looser. In another word, if a desirable mechanism exists, a market with a larger $\alpha$ can tolerate a larger data size disparity. A desirable mechanism is more likely to exist in a market that grows faster. When the market is not growing, there would not be such a desirable mechanism at all. On the other hand, if the market grows so fast such that there does not exist any competition between the agents, the mechanism always exists.

\appendix
\section*{Appendix}
\section{Algorithm}
\subsection{Finding a Desirable Mechanism}\label{algo}
%In the linear externality setting, we provide a mechanism that satisfies all the desirable properties. But this mechanism is not applicable to all valuation functions in the general setting, since the existence of a desirable mechanism depends on the agents’ actual valuation functions. We provide an algorithm, that given the agents’ valuations, computes whether such a mechanism exists, and outputs the one that optimizes revenue, if any.

 We provide an algorithm, that given the agents’ valuations, computes whether the mechanism that is simultaneously truthful, individually rational, efficient and weakly budget-balance exists, and outputs the one that optimizes revenue, if any. 

Since each agent can only under-report, according to the IR property, we must have:
$$ u_i(x(t_i,t_{-i}),t)\ge  u_i(
x(\emptyset,t_{-i}),t),\forall t,\forall i.$$
Equivalently, we get $\forall t, \forall i,$
\begin{gather*}
 u_i(x(\emptyset,t_{-i}),t) \le  v_i(x(t_i,t_{-i}),t)-p_i(t_i,t_{-i}),\\
p_i(t_i,t_{-i})\le v_i(x(t_i,t_{-i}),t)- u_i(
x(\emptyset,t_{-i}),t),\\
p_i(t) \le  v_i(x(t_i,t_{-i}),t)- u_i(
x(\emptyset,t_{-i}),t).
\end{gather*}
For simplicity, we define the upper bound of $p(t')$ as $$\overline{p(t)}\triangleq\{ v_i(x(t_i,t_{-i}),t)- u_i(
x(\emptyset,t_{-i}),t).$$
The IC property requires that $\forall t_i\ge t'_i, \forall t_{-i}, \forall i,$
\begin{gather*}
 u_i(
x(t_i,t_{-i}),t)\ge  u_i(
x(t_i',t_{-i}),t).
\end{gather*}
A little rearrangement gives:
\begin{align*}
&p_i(t_i,t_{-i})-p_i(t_i', t_{-i})\\
\le &v_i(x(t_i, t_{-i}),t) - v_i(x(t_i',t_{-i}),t)\\
\triangleq & Gap_i(t_i',t_i,t_{-i}).
\end{align*}
Note that the inequality correlations between the payments form a system of difference constraints. The form of update of the payments is almost identical to that of the shortest path problem. Therefore, we make use of this observation to design the algorithm. 

We assume that all the value functions are common knowledge, the efficient allocation is then determined because the mechanism always chooses the one that maximizes the social welfare. Thus it suffices to figure out whether there is a payment rule $p(t')$ which makes the mechanism IR, IC and weakly budget-balance. Since the valid data size for each agent is bounded in practice, we assume the mechanism only decides the payment functions on the data range $[0, D]$, and discretize the type space into intervals of length $\epsilon$,
% \footnote{If the agents’ types are not discretized, then the action space of the agents would be infinite. Since the utility function is general, all the points in the action space would influence the existence of the mechanism, it would take infinite time for the algorithm to calculate. Therefore we discretize the types to make this algorithmic characterization. Furthermore, when we are going to select the interval in practice, since the accuracy of the model is calculated according to the validation dataset, the interval $\epsilon$ could be set to the valid data amount which makes the accuracy increase to the next probable accuracy value.}
 which is also the minimal size of the data. Thus each agent's type is a multiple of $\epsilon$. Note that since the utility function is general, all the points in the action space would influence the properties and existence of the mechanism, thus it is necessary to enumerate all the points in the space. The exponential value function space, i.e., the exponential input space, determines that the complexity of our algorithm is exponential in D.

We give the following algorithmic characterization for the existence of a desirable mechanism.

\begin{algorithm2e}[!h]
    \SetKwInOut{Input}{input}
    \Input{Agents' valuation functions $v$.}

	Use the function $v_i$ to calculate all the $Gap_i(t_i',t_i,t_{-i})$ and $\overline{p_i(t_i, t_{-i})}$ for each $i$\;
	Initialize all $p_i^{max}(t_i, t_{-i})$ to be $\overline{p_i(t_i, t_{-i})}$ for each $i$\;
	\For{$i=1$ to $n$}{
	\For{$t_{-i} = ( \emptyset, \emptyset,\cdots, \emptyset )$ to $(D, D, \cdots, D)$ (increment = $\epsilon$ on each dimension)}{
    	Build an empty graph\;
    	For each $p_i(t_i, t_{-i})$, construct a vertex $V_{t_it_{-i}}$ and insert it into the graph\;
    	Construct a base vertex $VB_{t_{-i}}$ which denotes the payment zero into the graph\;
    	\For{$t_i = 0 $ to $D$ (increment = $\epsilon$)}{
    	    Add an edge from $VB_{t_{-i}}$ to $V_{t_it_{-i}}$ with weight $\overline{p(t_i,t_{-i})}$\;
    		\For{$t_i' = 0$ to $t_i$ (increment = $\epsilon$)}{
        		Add an edge with weight $Gap_i(t_i', t_i, t_{-i})$ from $V_{t_i't_{-i}}$ to $V_{t_it_{-i}}$\;
    		}
    	}
    	Use the Single-Source Shortest-Path algorithm to find the shortest path from $VB_{t_{-i}}$ to all the other vertices. These are the maximum solutions $p^{max}_i(t_i, t_{-i})$ for each payment case\;
    	\If{$\sum_{j=1}^n{p_j^{max}(t)}< 0$}{
    	    \Return There is no desirable mechanism.
    	}
	}
}
 	\Return {$p_i^{max}$} as the payment functions.
	\caption{Finding desirable mechanisms}
	\label{algo:existence}
\end{algorithm2e}
The following theorem proves the correctness of Algorithm \ref{algo:existence}.
\begin{theorem}
Taking agents' valuation functions as input, Algorithm \ref{algo:existence} outputs the answer of the decision problem of whether there exists a mechanism that guarantees IR, IC, efficiency and weak budget balance simultaneously, and specifies the payments that achieve maximal revenue if the answer is yes.
\end{theorem}
\begin{proof}
Suppose that there is a larger payment for agent $i$ such that $p_i(t')>p_i^{\max}(t')$ where $t'$ is the profile of reported types. In the process of our algorithm, the $p_i^{\max}(t')$ is the minimal path length from $VB_{-i}$ to $V_{t_i}{t_{-i}}$, denoted by $(VB_{-i}, V_{t_{i1}}{t_{-i}},V_{t_{i2}}{t_{-i}},\cdots, $ $V_{t_{ik}=t_i'}{t_{-i}})$. By the definition of edge weight, we have the following inequalities:
\begin{gather*}
p_i(t_{i1}, t_{-i})\le \overline{p_i(t_{i1},t_{-i})},\\
p_i(t_{i2},t_{-i}) - p_i(t_{i1},t_{-i})\le Gap_i(t_{i1}, t_{i2},t_{-i}),\\
\vdots\\
p_i(t_{ik},t_{-i}) - p_i(t_{i(k-1)},t_{-i})\le Gap_i(t_{i1}, t_{i2},t_{-i}).\\
\end{gather*}
Adding these inequalities together, we get $$p_i(t')\le \overline{p_i(t_{i1},t_{-i})} + \sum_{j=1}^{k-1} Gap_i(t_{ij}, t_{i(j+1)},t_{-i})=p_i^{\max}(t').$$ 

If $p_i(t')<p_i^{\max}(t')$ holds, this would violate at least 1 of the $k$ inequalities above. If the first inequality is violated, the mechanism would not be IR, by the definition of $\overline{p_i(t_{i1},t_{-i})}$. If any other inequality is violated, the mechanism would not be IC, by the definition of $Gap_i(t_{ij}, t_{i(j+1)},t_{-i})$. 

On the other hand, if we select $p_i^{\max}(t')$ to be payment of agent $i$, all the inequalities should be satisfied, otherwise the shortest path would be updated to a smaller length.

Therefore the $p_i^{\max}(t')$ must be the maximum payment for agent $i$. If the maximal payment sum up to less than 0, there would obviously be no mechanism that is IR, IC and weakly budget-balance under the efficient allocation function.
\end{proof}

\bibliographystyle{plainnat}
\bibliography{ref.bib}

\begin{thebibliography}{34}
\providecommand{\natexlab}[1]{#1}
\providecommand{\url}[1]{\texttt{#1}}
\expandafter\ifx\csname urlstyle\endcsname\relax
  \providecommand{\doi}[1]{doi: #1}\else
  \providecommand{\doi}{doi: \begingroup \urlstyle{rm}\Url}\fi

\bibitem[Abadi et~al.(2016)Abadi, Chu, Goodfellow, McMahan, Mironov, Talwar,
  and Zhang]{abadi2016deep}
Martin Abadi, Andy Chu, Ian Goodfellow, H~Brendan McMahan, Ilya Mironov, Kunal
  Talwar, and Li~Zhang.
\newblock Deep learning with differential privacy.
\newblock In \emph{Proceedings of the 2016 ACM SIGSAC Conference on Computer
  and Communications Security}, pages 308--318. ACM, 2016.

\bibitem[Ausubel(2004)]{ausubel2004efficient}
Lawrence~M Ausubel.
\newblock An efficient ascending-bid auction for multiple objects.
\newblock \emph{American Economic Review}, 94\penalty0 (5):\penalty0
  1452--1475, 2004.

\bibitem[Blumrosen and Feldman(2006)]{blumrosen2006implementation}
Liad Blumrosen and Michal Feldman.
\newblock Implementation with a bounded action space.
\newblock In \emph{Proceedings of the 7th ACM conference on Electronic
  commerce}, pages 62--71. ACM, 2006.

\bibitem[Blumrosen and Feldman(2013)]{blumrosen2013mechanism}
Liad Blumrosen and Michal Feldman.
\newblock Mechanism design with a restricted action space.
\newblock \emph{Games and Economic Behavior}, 82:\penalty0 424--443, 2013.

\bibitem[Blumrosen and Nisan(2002)]{blumrosen2002auctions}
Liad Blumrosen and Noam Nisan.
\newblock Auctions with severely bounded communication.
\newblock In \emph{The 43rd Annual IEEE Symposium on Foundations of Computer
  Science, 2002. Proceedings.}, pages 406--415. IEEE, 2002.

\bibitem[Blumrosen et~al.(2007)Blumrosen, Nisan, and
  Segal]{blumrosen2007auctions}
Liad Blumrosen, Noam Nisan, and Ilya Segal.
\newblock Auctions with severely bounded communication.
\newblock \emph{Journal of Artificial Intelligence Research}, 28:\penalty0
  233--266, 2007.

\bibitem[Brocas(2013)]{brocas2013optimal}
Isabelle Brocas.
\newblock Optimal allocation mechanisms with type-dependent negative
  externalities.
\newblock \emph{Theory and decision}, 75\penalty0 (3):\penalty0 359--387, 2013.

\bibitem[Chen et~al.(2018)Chen, Podimata, Procaccia, and
  Shah]{chen2018strategyproof}
Yiling Chen, Chara Podimata, Ariel~D Procaccia, and Nisarg Shah.
\newblock Strategyproof linear regression in high dimensions.
\newblock In \emph{Proceedings of the 2018 ACM Conference on Economics and
  Computation}, pages 9--26. ACM, 2018.

\bibitem[Chwe(1989)]{chwe1989discrete}
Michael Suk-Young Chwe.
\newblock The discrete bid first auction.
\newblock \emph{Economics Letters}, 31\penalty0 (4):\penalty0 303--306, 1989.

\bibitem[Cummings et~al.(2015)Cummings, Ioannidis, and
  Ligett]{cummings2015truthful}
Rachel Cummings, Stratis Ioannidis, and Katrina Ligett.
\newblock Truthful linear regression.
\newblock In \emph{Conference on Learning Theory}, pages 448--483, 2015.

\bibitem[David et~al.(2007)David, Rogers, Jennings, Schiff, Kraus, and
  Rothkopf]{david2007optimal}
Esther David, Alex Rogers, Nicholas~R Jennings, Jeremy Schiff, Sarit Kraus, and
  Michael~H Rothkopf.
\newblock Optimal design of english auctions with discrete bid levels.
\newblock \emph{ACM Transactions on Internet Technology (TOIT)}, 7\penalty0
  (2):\penalty0 12, 2007.

\bibitem[Dekel et~al.(2008)Dekel, Fischer, and Procaccia]{dekel2008incentive}
Ofer Dekel, Felix Fischer, and Ariel~D Procaccia.
\newblock Incentive compatible regression learning.
\newblock In \emph{Proceedings of the nineteenth annual ACM-SIAM symposium on
  Discrete algorithms}, pages 884--893. Society for Industrial and Applied
  Mathematics, 2008.

\bibitem[Deng and Pekec(2011)]{deng2011money}
Changrong Deng and Sasa Pekec.
\newblock Money for nothing: exploiting negative externalities.
\newblock In \emph{Proceedings of the 12th ACM conference on Electronic
  commerce}, pages 361--370. ACM, 2011.

\bibitem[Fang et~al.(2016)Fang, Tang, and Zuo]{fang2016digital}
Wenyi Fang, Pingzhong Tang, and Song Zuo.
\newblock Digital good exchange.
\newblock In \emph{Proceedings of the 2016 International Conference on
  Autonomous Agents \& Multiagent Systems}, pages 1277--1278. International
  Foundation for Autonomous Agents and Multiagent Systems, 2016.

\bibitem[Feigenbaum and Shenker(2004)]{feigenbaum2004distributed}
Joan Feigenbaum and Scott Shenker.
\newblock Distributed algorithmic mechanism design: Recent results and future
  directions.
\newblock In \emph{Current Trends in Theoretical Computer Science: The
  Challenge of the New Century Vol 1: Algorithms and Complexity Vol 2: Formal
  Models and Semantics}, pages 403--434. World Scientific, 2004.

\bibitem[Goldberg and Hartline(2003)]{goldberg2003envy}
Andrew~V Goldberg and Jason~D Hartline.
\newblock Envy-free auctions for digital goods.
\newblock In \emph{Proceedings of the 4th ACM conference on Electronic
  commerce}, pages 29--35. ACM, 2003.

\bibitem[Goldberg et~al.(2001)Goldberg, Hartline, and
  Wright]{goldberg2001competitive}
Andrew~V Goldberg, Jason~D Hartline, and Andrew Wright.
\newblock Competitive auctions and digital goods.
\newblock In \emph{Proceedings of the twelfth annual ACM-SIAM symposium on
  Discrete algorithms}, pages 735--744. Society for Industrial and Applied
  Mathematics, 2001.

\bibitem[Hu et~al.(2019)Hu, Niu, Yang, and Zhou]{hu2019fdml}
Yaochen Hu, Di~Niu, Jianming Yang, and Shengping Zhou.
\newblock Fdml: A collaborative machine learning framework for distributed
  features.
\newblock In \emph{Proceedings of the 25th ACM SIGKDD International Conference
  on Knowledge Discovery \& Data Mining}, pages 2232--2240. ACM, 2019.

\bibitem[Jehiel et~al.(1996)Jehiel, Moldovanu, and Stacchetti]{jehiel1996not}
Philippe Jehiel, Benny Moldovanu, and Ennio Stacchetti.
\newblock How (not) to sell nuclear weapons.
\newblock \emph{The American Economic Review}, pages 814--829, 1996.

\bibitem[Jehiel et~al.(1999)Jehiel, Moldovanu, and
  Stacchetti]{jehiel1999multidimensional}
Philippe Jehiel, Benny Moldovanu, and Ennio Stacchetti.
\newblock Multidimensional mechanism design for auctions with externalities.
\newblock \emph{Journal of economic theory}, 85\penalty0 (2):\penalty0
  258--293, 1999.

\bibitem[Kang et~al.(2019)Kang, Xiong, Niyato, Yu, Liang, and
  Kim]{kang2019incentive}
Jiawen Kang, Zehui Xiong, Dusit Niyato, Han Yu, Ying-Chang Liang, and Dong~In
  Kim.
\newblock Incentive design for efficient federated learning in mobile networks:
  A contract theory approach.
\newblock \emph{arXiv preprint arXiv:1905.07479}, 2019.

\bibitem[Kone{\v{c}}n{\`y} et~al.(2016)Kone{\v{c}}n{\`y}, McMahan, Yu,
  Richt{\'a}rik, Suresh, and Bacon]{konevcny2016federated}
Jakub Kone{\v{c}}n{\`y}, H~Brendan McMahan, Felix~X Yu, Peter Richt{\'a}rik,
  Ananda~Theertha Suresh, and Dave Bacon.
\newblock Federated learning: Strategies for improving communication
  efficiency.
\newblock \emph{arXiv preprint arXiv:1610.05492}, 2016.

\bibitem[McMahan et~al.(2017)McMahan, Moore, Ramage, Hampson, and
  y~Arcas]{mcmahan2017communication}
Brendan McMahan, Eider Moore, Daniel Ramage, Seth Hampson, and Blaise~Aguera
  y~Arcas.
\newblock Communication-efficient learning of deep networks from decentralized
  data.
\newblock In \emph{Artificial Intelligence and Statistics}, pages 1273--1282,
  2017.

\bibitem[Meir et~al.(2009)Meir, Procaccia, and
  Rosenschein]{meir2009strategyproof}
Reshef Meir, Ariel~D Procaccia, and Jeffrey~S Rosenschein.
\newblock Strategyproof classification with shared inputs.
\newblock In \emph{Twenty-First International Joint Conference on Artificial
  Intelligence}, 2009.

\bibitem[Meir et~al.(2012)Meir, Procaccia, and Rosenschein]{meir2012algorithms}
Reshef Meir, Ariel~D Procaccia, and Jeffrey~S Rosenschein.
\newblock Algorithms for strategyproof classification.
\newblock \emph{Artificial Intelligence}, 186:\penalty0 123--156, 2012.

\bibitem[Myerson and Satterthwaite(1983)]{myerson1983efficient}
Roger~B Myerson and Mark~A Satterthwaite.
\newblock Efficient mechanisms for bilateral trading.
\newblock \emph{Journal of economic theory}, 29\penalty0 (2):\penalty0
  265--281, 1983.

\bibitem[Nisan and Ronen(2001)]{nisan2001algorithmic}
Noam Nisan and Amir Ronen.
\newblock Algorithmic mechanism design.
\newblock \emph{Games and Economic behavior}, 35\penalty0 (1-2):\penalty0
  166--196, 2001.

\bibitem[Nisan et~al.(2008)Nisan, Roughgarden, Tardos, and
  Vazirani]{nisan2008agt}
Noam Nisan, Tim Roughgarden, Eva Tardos, and Vijay~V Vazirani.
\newblock \emph{Algorithmic game theory}.
\newblock Cambridge University Press, 2008.

\bibitem[Perote and Perote-Pena(2004)]{perote2004strategy}
Javier Perote and Juan Perote-Pena.
\newblock Strategy-proof estimators for simple regression.
\newblock \emph{Mathematical Social Sciences}, 47\penalty0 (2):\penalty0
  153--176, 2004.

\bibitem[Redko and Laclau(2019)]{redko2019fair}
Ievgen Redko and Charlotte Laclau.
\newblock On fair cost sharing games in machine learning.
\newblock In \emph{Thirty-Third AAAI Conference on Artificial Intelligence},
  2019.

\bibitem[Shokri and Shmatikov(2015)]{shokri2015privacy}
Reza Shokri and Vitaly Shmatikov.
\newblock Privacy-preserving deep learning.
\newblock In \emph{Proceedings of the 22nd ACM SIGSAC conference on computer
  and communications security}, pages 1310--1321. ACM, 2015.

\bibitem[Smith et~al.(2017)Smith, Chiang, Sanjabi, and
  Talwalkar]{smith2017federated}
Virginia Smith, Chao-Kai Chiang, Maziar Sanjabi, and Ameet~S Talwalkar.
\newblock Federated multi-task learning.
\newblock In \emph{Advances in Neural Information Processing Systems}, pages
  4424--4434, 2017.

\bibitem[Takabi et~al.()Takabi, Hesamifard, and Ghasemi]{takabi2016privacy}
Hassan Takabi, Ehsan Hesamifard, and Mehdi Ghasemi.
\newblock Privacy preserving multi-party machine learning with homomorphic
  encryption.

\bibitem[Yonetani et~al.(2017)Yonetani, Naresh~Boddeti, Kitani, and
  Sato]{yonetani2017privacy}
Ryo Yonetani, Vishnu Naresh~Boddeti, Kris~M Kitani, and Yoichi Sato.
\newblock Privacy-preserving visual learning using doubly permuted homomorphic
  encryption.
\newblock In \emph{Proceedings of the IEEE International Conference on Computer
  Vision}, pages 2040--2050, 2017.

\end{thebibliography}

\end{document}